\documentclass[entropy,article,accept,oneauthor,pdftex,10pt,a4paper]{mdpi} 

\firstpage{1} 
\lastpage{x} 
\doinum{10.3390/------}
\pubvolume{xx}
\pubyear{2016}
\externaleditor{Academic Editor: name}
\history{Received: date ; Accepted: date ; Published: date}
\pdfoutput=1

\usepackage{amssymb}
\usepackage{graphicx}
\usepackage{bm}
\usepackage{hyperref}
\usepackage{color}

\newcommand\e{\operatorname{e}}

 \theoremstyle{mdpi}
 \newcounter{thm}
 \newcounter{ex}
 \newcounter{re}
 \newtheorem{Theorem}[thm]{Theorem}

 \theoremstyle{mdpidefinition}

 \newtheorem{Definition}[thm]{Definition}
\graphicspath{{./figures/}}

\Title{A Cost/ Speed/ Reliability Tradeoff to Erasing$^\dagger$}

\Author{Manoj Gopalkrishnan}

\address[1]{%
$^{1}$ \quad Tata Institute of Fundamental Research; manoj.gopalkrishnan@gmail.com}

\corres{Correspondence: manoj.gopalkrishnan@gmail.com; Tel.: +91-22-2278-2738}

\firstnote{This paper is an extended version of the paper entitled ''Cost / speed / reliability tradeoff in erasing a~bit'', presented at the 14th International Conference on Unconventional and Natural Computation, Auckland, New Zealand}

\abstract{We present a KL-control treatment of the fundamental problem of erasing a bit. We introduce notions of \textbf{reliability} of information storage via a reliability timescale $\tau_r$, and \textbf{speed} of erasing via an erasing timescale $\tau_e$.  Our problem formulation captures the tradeoff between speed, reliability, and the Kullback-Leibler (KL) cost required to erase a bit. We show that rapid erasing of a reliable bit costs at least $\log 2 - \log\left(1 - \operatorname{e}^{-\frac{\tau_e}{\tau_r}}\right) > \log 2$, which goes to  $\frac{1}{2} \log\frac{2\tau_r}{\tau_e}$ when $\tau_r>>\tau_e$.}


\begin{document}

\section{Motivation}
Biological systems are remarkably ordered at multiple scales and dimensions, from the spatial order witnessed in the packing of DNA inside the nucleus, the arrangement of cells to form tissues, and organs, and whole organisms, to the temporal order witnessed in the execution of various cellular processes. Superficially such order might appear to violate the second law of thermodynamics which requires an increase in disorder with overwhelmingly high probability. In fact, there is no violation since biological systems expend energy to bring about and maintain this increase in order.

We would like to understand this ``energy to order'' conversion quantitatively. What are the fundamental limits to this conversion? Order can be measured in terms of information, by counting the number of bits required to describe that order. From this point of view, understanding how much energy is required to create order becomes an instance of the investigation of the connection between information processing and thermodynamics. The basic information processing operation that increases order is the operation of ``erasing'' or resetting a bit to state 0. To fix ideas, imagine erasing random chalk marks from a blackboard, to leave it in a neat and ordered state.

Szilard~\cite{Szilard1929} and later Landauer~\cite{landauer61irreversibility} have argued from the second law of thermodynamics that erasing at temperature $T$ requires at least $k_B T\log 2$ units of energy, where $k_B$ is Boltzmann's constant. The \textbf{Szilard engine} is a simple illustration of this result. Imagine a single molecule of ideal gas in a cylindrical vessel. If this molecule is in the left half of the vessel, think of that as encoding the bit ``0,'' and the bit ``1'' otherwise. Erasing this Brownian bit corresponds to ensuring that the molecule lies on the left half, for example by compressing the ideal gas to half its volume. For a heuristic analysis we may use the ideal gas law $P V = k_B T$, integrating the expression $dW = - P dV$ for work from limits $V$ to $V/2$ to obtain $W = k_B T\log 2$. More rigorous and general versions of this calculation are known, which also clarify why this is a lower bound~\cite{esposito2011second,gopalkrishnan2013hot,reeb2013proving}.

In practice, one finds that both man-made and biological instrumentation often require energy substantially more than $k_B T\log 2$ to perform erasing~\cite{laughlin1998metabolic,mudge2001power}. John von Neumann remarked on this large gap in his 1949 lectures at the University of Illinois~\cite{neumann66theory}. (Bennett~\cite{bennett82thermodynamics} has remarked that DNA polymerases come close to the $k_B T \log 2$ bound. To copy a single base, a DNA polymerase hydrolyzes a triphosphate molecule to a monophosphate, which provides close to $18 k_B T$ at temperature $T=300 K$. Note that this is still almost two orders of magnitude away from $k_BT\log 2$. Further, it is not clear whether the comparison is valid at all since copying and erasing are different operations.)

How does one explain this large gap? Note that the result of $k_B T\log 2$ holds only in the isothermal limit, which takes infinite time. In practice, we want erasing to be performed rapidly, say in time $\tau_e$, which requires extra entropy production. For intuition, suppose one wants to compress a gas in finite time $\tau_e$. The gas heats up, and pushes back, increasing the work required.

Several groups~\cite{aurell2012refined,diana2013finite,zulkowski2014optimal} have recognized that rapid erasing requires entropy production which pushes up the cost of erasing beyond $k_B T\log 2$, and have obtained bounds for this problem. A grossly oversimplified, yet qualitatively accurate, sketch of these various results is obtained by considering the energy cost of compressing the Szilard engine rapidly. Specializing a result from finite-time thermodynamics~\cite{salamon1981finite} to the case of the Szilard engine, one obtains an energy cost $\left(1 + \frac{ k_B \log 2}{\sigma \tau_e - k_B\log 2}\right)k_B T\log 2$ where $\sigma$ is the coefficient of heat conductivity of the vessel. 

The bounds obtained by such considerations depend on technological parameters like the heat conductivity $\sigma$, and not just on fundamental constants of physics and the requirement specifications of the problem. If one varies over the technological parameters as well, e.g. allowing $\sigma\to\infty$, the energy cost tends to $k_B T\log 2$. Does there exist a more fundamental analysis for the cost of erasing that is independent of technological parameters, and improves on $k_B T\log 2$? This is the open question we address in this paper.

\textbf{Our contribution:} We follow up on von Neumann's suggestion~\cite{neumann66theory} that the gap was ``due to something like a desire for reliability of operation.'' Swanson~\cite{swanson1960physical} and Alicki~\cite{alicki2014information} have also looked into issues of reliability. We introduce the notion of ``reliability timescale'' $\tau_r$, and explicitly consider the three-way trade-off between speed, reliability, and cost.

The other novelty of our approach is in bringing the tools of Kullback-Leibler (KL) control~\cite{todorov2009efficient,fleming1982optimal} to bear on the problem of erasing a bit. The intuitive idea is that the control can reshape the dynamics as it pleases, but pays for the deviation from the uncontrolled dynamics. The cost of reshaping the dynamics is a relative entropy or KL divergence between the controlled and the uncontrolled dynamics, expressed as measures on path space. 
%
%
%
%
%
%

We find the optimal control for rapid erasing of a reliable bit, and argue that it requires cost of at least $\log 2 - \log\left(1 - \operatorname{e}^{-\frac{\tau_e}{\tau_r}}\right) > \log 2$, which goes to  $\frac{1}{2} \log\frac{2\tau_r}{\tau_e}$ when $\tau_r>>\tau_e$. Importantly, our answer does not depend on any technological parameters, but only on the requirement specifications $\tau_r$ and $\tau_e$ of the problem.

\section{The Erasing Problem}

As a model of a bit, consider a two-state continuous-time Markov chain with states $0$ and $1$ and the \textbf{passive} or uncontrolled dynamics given by transition rates $k_{01}$ from state $0$ to state $1$ and $k_{10}$ from state $1$ to state $0$.
\begin{center}
\begingroup%
  \makeatletter%
  \providecommand\color[2][]{%
    \errmessage{(Inkscape) Color is used for the text in Inkscape, but the package 'color.sty' is not loaded}%
    \renewcommand\color[2][]{}%
  }%
  \providecommand\transparent[1]{%
    \errmessage{(Inkscape) Transparency is used (non-zero) for the text in Inkscape, but the package 'transparent.sty' is not loaded}%
    \renewcommand\transparent[1]{}%
  }%
  \providecommand\rotatebox[2]{#2}%
  \ifx\svgwidth\undefined%
    \setlength{\unitlength}{103.28046875bp}%
    \ifx\svgscale\undefined%
      \relax%
    \else%
      \setlength{\unitlength}{\unitlength * \real{\svgscale}}%
    \fi%
  \else%
    \setlength{\unitlength}{\svgwidth}%
  \fi%
  \global\let\svgwidth\undefined%
  \global\let\svgscale\undefined%
  \makeatother%
  \begin{picture}(1,0.29361795)%
    \put(0,0){\includegraphics[width=\unitlength]{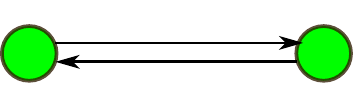}}%
    \put(0.4045713,0.22095477){\color[rgb]{0,0,0}\makebox(0,0)[lb]{\smash{$k_{01}$}}}%
    \put(0.4200631,0.01956142){\color[rgb]{0,0,0}\makebox(0,0)[lb]{\smash{$k_{10}$}}}%
    \put(0.057126,0.1134737){\color[rgb]{0,0,0}\makebox(0,0)[lb]{\smash{$0$}}}%
    \put(0.87431826,0.1134737){\color[rgb]{0,0,0}\makebox(0,0)[lb]{\smash{$1$}}}%
  \end{picture}%
\endgroup%

\end{center}
The transition rates $k_{01}$ and $k_{10}$ model spontaneous transitions between the states when no one is looking at the bit or trying to erase it. The time independence of these rates represents the physical fact that the system is not being driven.

Such finite Markov chain models often arise in physics by ``coarse-graining.'' For example, for the case of the Szilard engine, the transition rate $k_{10}$ models the rate at which the molecule enters the left side, conditioned on it currently being on the right side.

Apart from their importance in approximating the behavior of real physical systems, finite Markov chains are also important to thermodynamics from a logical point of view. They may be viewed as finite models of a mathematical theory of thermodynamics. The terms ``theory'' and ``model'' are to be understood in their technical sense as used in mathematical logic. We develop this remark no further here since doing so would take us far afield.

Suppose the distribution at time $t$ is $(p_0(t),p_1(t))$ with $p_1(t) = 1-p_0(t)$. Then the time evolution of the bit is described by the ODE

\begin{align}\label{eqn:passive}
\dot{p}_0(t) = -k_{01}p_0(t) + k_{10}(1-p_0(t)).
\end{align}
Setting $\pi_0 = k_{10}/(k_{01}+k_{10})$ and the \textbf{reliability timescale} $\tau_r := 1/(k_{01} + k_{10})$, this admits the solution

\begin{align}\label{eqn:soln}
	p_0(t) =  \pi_0 + \e^{-t/\tau_r}( p_0(0) - \pi_0)
\end{align}
Here $\tau_r$ represents the time scale on which memory is preserved. The smaller the rates $k_{01}$ and $k_{10}$, the larger is the value of $\tau_r$, and the slower the decay to equilibrium, so that the system remembers information for longer.

Fix a \textbf{required erasing time} $\tau_e$. Fix $p(0) = \pi_0$. We want to control the dynamics with transition rates $u_{01}(t)$ and $u_{10}(t)$ to achieve $p(\tau_e) = (1,0)$, where

\begin{align}\label{eqn:control}
\dot{p}_0(t) &= -u_{01}(t)p_0(t) + u_{10}(t)(1-p_0(t))
\end{align}
We want to find the cost of the optimal protocol $u_{01}^*(t)$ and $u_{10}^*(t)$ to achieve this objective, according to a cost function which we introduce next. In particular, when $k_{01} = k_{10}$, the equilibrium distribution $\pi = (\pi_0,1-\pi_0)$ takes the value $(1/2,1/2)$ and we can interpret this task as erasing a bit of reliability $\tau_r=1/(k_{01}+k_{10})$ in time $\tau_e$.

\subsection{Kullback Leibler Cost}

\label{def:2stateenergy} Define the \textbf{path space} $\mathcal{P} := \{0,1\}^{[0,\tau_e]}$ of the two-state Markov chain. This is the set of all paths in the time interval $[0,\tau_e]$ that jump between states $0$ and $1$ of the Markov chain. Each path can also be succinctly described by its initial state, and the times at which jumps occur. We can also effectively think of the path space as the limit as $h\to 0$ of the space $\mathcal{P}_h:=\{0,1\}^{\{0,h,2h,\dots, Nh=\tau_e\}}$ corresponding to the discrete-time Markov chain that can only jump at clock ticks of $h$ units. 

Once the rates $u_{01}(t),u_{10}(t)$ and the initial distribution $p(0)=p$ for the Markov chain are fixed, there is a unique probability measure $\mu_{u,p}$ on path space which intuitively assigns to every path the probability of occurrence of that path according to the Markov chain evolution (Equation~\ref{eqn:control}) with initial conditions $p$. 

For pedagogic reasons, we first describe the discrete-time measure $\mu_{u,p}^h$ for a single path $i = (i_0,i_1,\dots, i_N)\in\mathcal{P}_h$. First we describe the transition probabilities of the discrete-time Markov chain.
For $a,b\in\{0,1\}$ with $a\neq b$, for all times $t$,
define $u_{aa}^h(t):= 1 - h u_{ab}(t)$ and 
$u_{ab}^h(t):= h u_{ab}(t)$ as the probability of jumping to $a$ 
and to $b$ respectively in the time step $t$, conditioned on being in state $a$.
Then the probability of the path $i$ under control $u$ is given by:

\[
\mu_{u,p}(i) := p_{i_0}\prod_{j=0}^{N-1} u_{i_j,i_{j+1}}^h(jh)
\]

\begin{figure}[H]
\centering
\includegraphics[scale = 0.40,trim=0 8cm 8 5cm,clip=true]{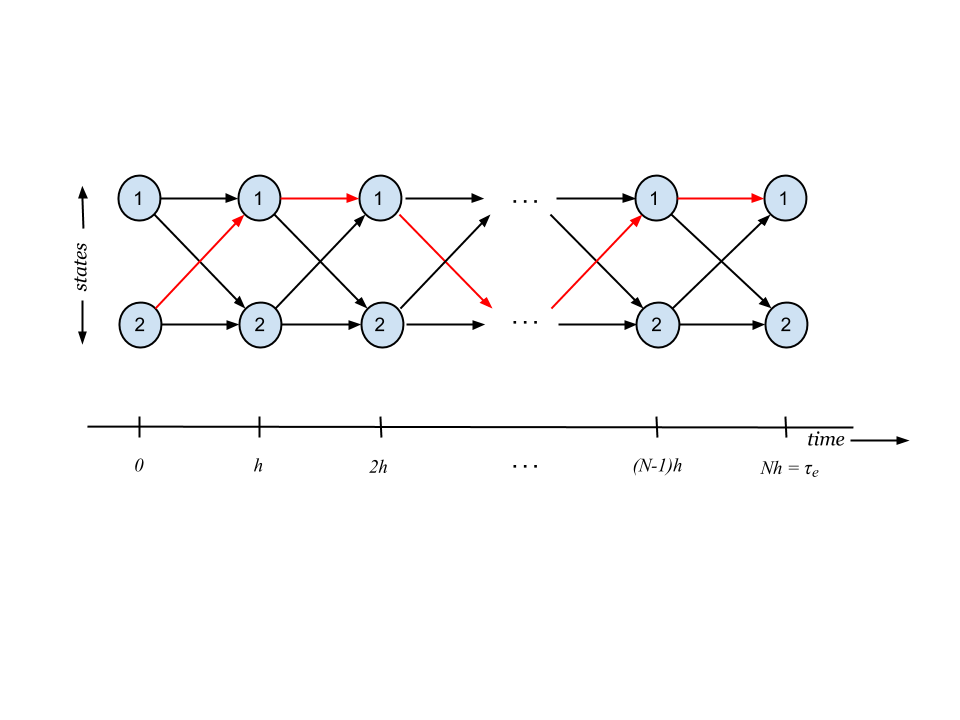}
\caption{\label{fig1} The discrete-time path space $\mathcal{P}_h$. A specific path is labeled in red.}
\end{figure}

We describe the continous-time case now. We could obtain the measure $\mu_{u,p}$ from $\mu_{u,p}^h$ by sending $h\to 0$, but it can also be described more directly. Fix $i_0\in\{0,1\}$, and consider the set of paths $\mathcal{S} = \mathcal{S}_{i_0,t_1,t_2,\dots,t_n}$ starting at $i_0$ with jumps occurring around times $t_1<t_2<\dots<t_n$ within infinitesimal intervals $dt_1,dt_2,\dots, dt_n$ and leading to the trajectory $(i_0,i_1,\dots,i_n)\in\{0,1\}^{n+1}$. Setting $t_0=0$:

\[
\mu_{u,p}(\mathcal{S}) := p_{i_0}\prod_{j=0}^{n-1} \e^{-\int_{t_j}^{t_{j+1}} u_{i_ji_{j+1}}(s) ds} u_{i_ji_{j+1}}(t_{j+1}) dt_{j+1}
\]
where $p_{i_0}$ is the probability of starting at $i_0$, $\e^{-\int_0^{t_1} u_{i_0i_1}(s) ds}$ is the probability of not jumping in the time interval $(0,t_1)$, $u_{i_0i_1}(t_1) dt_1$ is the probability of jumping from $i_0$ to $i_1$ in the interval $(t_1,t_1+dt_1)$ and so on. This is the well-known Feynam-Kac formula for this Markov chain.

Specializing to $u_{01}(t) = k_{01}$ and $u_{10}(t) = k_{10}$, we obtain the probability measure $\mu_{k,p}$ induced on $\mathcal{P}$ by the passive dynamics (Equation~\ref{eqn:passive}) with initial conditions $p$.

We declare the Kullback Leibler (KL) cost $D(\mu_{u,p}\,\|\, \mu_{k,p})$ as the cost for implementing the control $u$. More generally, for a physical system with path space $\mathcal{P}$, passive dynamics corresponding to a measure $\nu$ on $\mathcal{P}$, and a controlled dynamics with a control corresponding to a measure $\mu$ on $\mathcal{P}$, we declare $D(\mu\|\nu)$ as the cost for implementing the control. This cost function has been widely used in control theory~\cite{fleming1982optimal,kappen2005path,kappen2005linear,theodorou2011iterative,theodorou2012relative,stulp2012reinforcement,dvijotham2011unified,todorov2009efficient,kappen2012optimal,van2008graphical,wiegerinck2012stochastic,horowitz2014efficient}. In Section~\ref{Sec:PSSLP} we will explore some other interpretations of this cost function.


\section{Solution to the Erasing Problem}

Out of all controls $u(t)$, we want to find a control $u^*(t)$ that starts from 

\[
p(0)=\pi = \left(\frac{k_{10}}{(k_{01}+k_{10}}, \frac{k_{01}}{k_{01}+k_{10}}\right)
\]
and achieves $p(\tau_e)=(1,0)$ while minimizing the relative entropy $D(\mu_{u^*,\pi}\,\|\, \mu_{k,\pi})$.

\begin{align*}
	&\text{To find: } u^* = \arg\inf_u D(\mu_{u,\pi}\,\|\, \mu_{k,\pi})\
	\\&\text{Subject to: }\mu_{u,\pi}(\tau_e) = (1,0)
\end{align*}

Question~\ref{q:2erase} can be described within the framework of a well-studied problem in optimal control theory that has a closed-form solution~\cite{fleming1982optimal,todorov2009efficient,dupuis2011weak}. Following Todorov~\cite{todorov2009efficient}, we introduce the \textbf{optimal cost-to-go} function $v(t) = (v_0(t),v_1(t))$. We intend $v_i(t)$ to denote the expected cumulative cost for starting at state $i$ at time $t<\tau_e$, and reaching a distribution close to $(1, 0)$ at time $\tau_e$. 

To discourage the system from being in state $1$ at time $\tau_e$, define $v_1(\tau_e) = +\infty$ and $v_0(\tau_e) = 0$. Suppose the control performs actions $u_{01}(t)$ and $u_{10}(t)$ at time $t$. Fix a small time $h>0$. Define the transition probability $u_{ij}^h(t)$ as the probability that a trajectory starting in state $i$ at time $t$ will be found in state $j$ at time $t+h$. When $i\neq j$, $u_{ij}^h(t) \approx h u_{ij}(t)$, whereas $u_{ii}^h(t) \approx 1 - u_{ij}^h(t)$ ignoring terms of size $O(h^2)$. We define $k_{ij}^h$ similarly.

Let ``$\log$'' denote the natural logarithm. To derive the law satisfied by the optimal cost-to-go $v(t)$, we approximate $v(t)$ by the backward recursion relations:
\begin{equation}
\begin{aligned}\label{eqn:togo}
	v_0(t) &= \min_{u_{01}(t)} \mathop{\mathbb{E}} \left[ v_i(t+h) + \log\frac{u_{0i}^h(t)}{k_{0i}^h}\right]\
\\v_1(t) &= \min_{u_{10}(t)} \mathop{\mathbb{E}} \left[ v_i(t+h) + \log\frac{u_{1i}^h(t)}{k_{1i}^h}\right]
\end{aligned}
\end{equation}
where the first expectation is over $i\sim_\text{law} (u_{00}^h(t), u_{01}^h(t))$, and the second is over $i\sim_\text{law} (u_{10}^h(t),u_{11}^h(t))$, and the approximation ignores terms of size $O(h^2)$. As $h\to 0$ the second terms $\mathop{\mathbb{E}}\log\frac{u_{ji}^h(t)}{k_{ji}^h}$ approach the relative entropy cost in path space over the time interval $(t,t+h)$.

In words, Equation~\ref{eqn:togo} says that the cost-to-go from state $0$ at time $t$ equals the cost of the control $u(t)$ plus the expected cost-to-go in the new state $i$ reached at time $t+h$. The cost of the control is measured by relative entropy of the control dynamics relative to the passive dynamics, over the time interval $(t,t+h)$.

Define the \textbf{desirability} $z_0(t) = \e^{-v_0(t)}$ and $z_1(t) = \e^{-v_1(t)}$. Define 

\begin{align*}
G_0[z](t) &=  k_{00}^h z_0(t) + k_{01}^h z_1(t),\ 
\\G_1[z](t) &=  k_{10}^h z_0(t) + k_{11}^h z_1(t).
\end{align*}
We can rewrite Equation~\ref{eqn:togo} as

\begin{equation}\label{eqn:desirability}
\begin{aligned}
\log z_0(t) &= \log G_0[z](t+h) - \min_{u_{01}(t)} 
\mathop{\mathbb{E}} 
\left[ \log\frac{u_{0i}^h(t)G_0[z](t+h)}{k_{0i}^h z_i(t+h)}\right]\
\\\log z_1(t) &= \log G_1[z](t+h) - \min_{u_{10}(t)} 
\mathop{\mathbb{E}} 
\left[ \log\frac{u_{1i}^h(t)G_1[z](t+h)}{k_{1i}^h z_i(t+h)}\right]
\end{aligned}
\end{equation}

Since the last term is the relative entropy of $(u_{j0}^h(t),u_{j1}^h(t))$ relative to the probability distribution $(k_{j0}^h z_0(t+h)/G_j[z](t+h),k_{j1}^h z_1(t+h)/G_j[z](t+h))$, its minimum value is $0$, and is achieved by the protocol $u^*$ given by:

\begin{align}\label{eqn:gibbs}
\frac{u^*_{ji}(t)}{k_{ji} }=\lim_{h\to 0}\frac{ \e^{-v_i(t+h)}}{G_j[z](t+h)}=\frac{\e^{-v_i(t)}}{\e^{-v_j(t)}}
\end{align}
when $i\neq j$. 



It remains to solve for $z(t)$ and the optimal cost. From Equation~\ref{eqn:desirability}, at the optimal control $u^*$ the desirability $z(t)$ must satisfy the equation $-\log z(t) = -\log G[z](t+h) + 0$, so that:
\begin{align*}
\left(\begin{array}{c}
z_0(t)\
\\z_1(t)
\end{array}\right) = 
\left(\begin{array}{cc}
1-k_{01}h & k_{01} h\
\\k_{10} h & 1-k_{10} h
\end{array}\right) \left(\begin{array}{c}
z_0(t+h)\
\\z_1(t+h)
\end{array}\right)
\end{align*}
which simplifies to $\frac{dz}{dt} = - K z$ in the limit $h\to 0$, where $K=\tiny\left(\begin{array}{cc}
-k_{01} & k_{10} \
\\k_{10} & -k_{10} 
\end{array}\right)$ is the infinitesimal generator  of the Markov chain. This equation has the formal solution 
$z(\tau_e - t) = \e^{K t} z(\tau_e)$ where $z(\tau_e) = \tiny\left(\begin{array}{c}1\ \\0\end{array}\right)$. In the symmetric case $k_{01} = k_{10}$,
\[
	z(t) = \left(\begin{array}{c}1/2\\ 1/2\end{array}\right) + \e^{-\frac{\tau_e - t}{\tau_r}}\left(\begin{array}{c}
	1/2 \\ -1/2 \end{array}\right)
\]
where $\tau_r = 1/(k_{01}+k_{10})$. Substituting $t=0$ and taking logarithms, we find the cost-to-go function at time $0$:

\[
v(0) = 	\left(\begin{array}{c} \log 2 \\  \log 2 \end{array}\right) - \left(\begin{array}{c} \log\left(1 + \e^{-\tau_e/\tau_r}\right)\\ \log\left(1 - \e^{-\tau_e/\tau_r}\right)\end{array}\right)
\]

When $\nu(0) = (1/2,1/2)$ with $k_{01} = k_{10}$, the cost $C_\text{erase}(\tau_r,\tau_e,T)$ required for erasing a bit of reliability $\tau_r=1/(k_{01}+k_{10})$ in time $\tau_e$ is at least:

\begin{align}\label{eqn:erasingcost}
 \log 2 - \frac 1 2 \log \left(1  - \e^{-2\tau_e/\tau_r} \right)
\end{align}
Note that $C_\text{erase}\geq  \log 2$ with equality when $\tau_e/\tau_r\to \infty$, since ${1  - \e^{-2\tau_e/\tau_r} \leq 1}$.

From Equation~\ref{eqn:erasingcost},  $C_\text{erase} \geq \frac{1}{2} \log\frac{2\tau_r}{\tau_e}$ when $\tau_r >> \tau_e$. 


\section{Interpreting the KL cost}\label{Sec:PSSLP}
One motivation for our cost function comes from the field of KL control theory. We now compare other possible meanings to this cost function.

\subsection{Path space Szilard-Landauer correspondence}
The correspondence between information and thermodynamics was revealed in the work of Szilard, and clarified by Landauer. More rigorous and general treatments of this correspondence have been worked out recently~\cite{esposito2011second, gopalkrishnan2013hot, reeb2013proving}. We first recall this result, and then show how our cost function is a formal extension of this result.

Consider a physical system with finite state space $S$ and \textbf{energy} $E:S\to\mathbb{R}$. (More general state spaces $S$ can be handled by replacing the sum by an appropriate integral. For our present purposes, it suffices to assume $S$ is finite.) Define the \textbf{Gibbs distribution} $\pi$ at temperature $T$ by 

\[
\pi(i) = \frac{\e^{-E_i/k_BT}}{\sum_{j\in S} \e^{-E_j/k_BT}}
\]
for all $i\in S$. Define the \textbf{free energy} 

\[
F(p) := \sum_{i\in S} p_i E_i - k_B T \sum_{i\in S} p_i \log \frac{1}{p_i}
\] where $p$ is a probability distribution. 

Define the \textbf{relative entropy} $D(p\|q)=\sum_{i\in S} p_i \log\frac{p_i}{q_i}$ with Euler's constant for the base of the logarithm. Following Jaynes~\cite{jaynes1957information}, assume that equilibrium $\pi$ corresponds to a maximally uninformative state of the system, so that we have zero information about the system when it is at equilibrium. Recall that a nat is the unit of information when logarithms are taken to the base of Euler's constant. $1$ bit $= \log 2$ nats. Then the relative entropy $D(p\|\pi)$ has an axiomatic identification with the amount of information in nats that we know about the system when it is in a nonequilibrium state $p$~\cite{gopalkrishnan2013hot}.

The following identity is easily verified:
\begin{align}\label{SLP}
F(p) - F(\pi) = k_B T D(p \| \pi).
\end{align}
The conceptual significance of this simple identity is that it supplies a dictionary between thermodynamics and information theory~\cite{gopalkrishnan2013hot}. In particular, erasing a bit corresponds to increasing relative entropy which in turn corresponds --- via the identity --- to increasing available free energy $F(p) - F(\pi)$ by $k_B T \log 2$, recovering the classical result of Szilard as an alternative statement of the second law of thermodynamics. In the other direction, charging a battery corresponds to increasing available free energy which in turn corresponds --- via Identity~\ref{SLP} --- to erasing of information. This relates the energy efficiency of charging a battery to the energy required to erase a bit.

Now consider our cost function $D(\mu\|\nu)$. The relative entropy $D(\mu\|\nu)$ counts the number of nats erased by the control in path space, relative to the passive dynamics. Since the Szilard-Landauer principle asserts that erasing one bit requires at least $k_B T\log 2$ units of energy, our cost function may be viewed as a \textbf{Path Space Szilard-Landauer Principle}, formally extending Identity~\ref{SLP} to path space. 

\subsection{Thermodynamic interpretation} We wish to compare the cost $D(\mu\|\nu)$ with the usual thermodynamic expected work $\Delta W$. We will quickly outline how thermodynamic quantities can be defined for a two-state Markov chain. 

\subsubsection{Thermodynamics on a two-state Markov chain}
The ideas we present here are well-known in the nonequilibrium thermodynamics community, for example see Propp's thesis~\cite{propp1985thermodynamic}. The construction can be carried out more generally, but the generalization is not necessary for our present purposes.

\begin{enumerate}
\item Consider again the two-state continuous-time Markov chain with passive dynamics given by transition rates $k_{01}$ and $k_{10}$.
\begin{center}

\end{center}
Let $E_0$ and $E_1$ denote the \textbf{internal energy} of states ``0'' and ``1'' respectively. Then the equilibrium distribution is given by $\pi_0 \propto \e^{-E_0/k_BT}$ and $\pi_1 \propto \e^{-E_1/k_BT}$. We also have $k_{01} \pi_0 = k_{10} \pi_1$ from detailed balance. Together this yields 

\begin{align}\label{eqn:energy}
E_0 - E_1 = k_B T\log \frac{k_{01}}{k_{10}}.
\end{align}

\item Now consider the same two-state system with a control applied to it by means of a field of potential $\phi(t) = (\phi_0(t), \phi_1(t))$ so that the potential energy in state $i$ becomes $E_i + \phi_i(t)$. The transition rates due to the control become $u_{01}(t)$ and $u_{10}(t)$. By a reasoning similar to how we derived Equation~\ref{eqn:energy}, we get 

\[
E_0 + \phi_0 - E_1 - \phi_1 = k_BT\log\frac{u_{01}}{u_{10}}.
\] 
Combining with Equation~\ref{eqn:energy} this yields 

\begin{align}\label{eqn:potential}
\phi_0 - \phi_1 = k_B T\log \frac{u_{01}k_{10}}{u_{10}k_{01}}
\end{align}

\item Given a distribution $p = (p_0,p_1)$ on the states, we can define the following thermodynamic quantities:
\begin{itemize}
\item \textbf{Expected internal energy} $E(p) = p_0 E_0 + p_1 E_1$.
\item \textbf{Entropy} $S(p) = - p_0\log p_0 - p_1\log p_1$.
\item \textbf{Nonequilibrium free energy} $F(p) = E(p) - k_B T S(p)$.
\end{itemize}

\item Given a transition from state $i$ to state $j$ in the presence of the control field, we can define the following thermodynamic quantities:
\begin{itemize}
\item \textbf{Heat} dissipated $Q_{ij}(t) = E_i + \phi_i(t) - E_j - \phi_j(t)$.
\item \textbf{Work} done by the control $W_{ij}(t) = \phi_i(t) - \phi_j(t)$. This expression for work can be traced back to Sekimoto~\cite{sekimoto1997kinetic}, and is commonly employed in the field of Stochastic Thermodynamics to describe the work done by switching on a control field~\cite{SeifertBook}. 
\item \textbf{Entropy increase} of the system $S_{ij}(t) = \log\frac{p_i(t)}{p_j(t)}$.
\end{itemize}
The first law of thermodynamics manifests as $W_{ij} = Q_{ij} + E_i - E_j$.

\item Suppose the system is described at time $t$ by a distribution $p(t) =(p_0(t), p_1(t))$. Define the \textbf{Current} $J_{ij}(t) = p_i(t) u_{ij}(t) - p_j(t) u_{ji}(t)$ so that $\dot{p}_0(t) = -J_{01}(t)$.

\item\label{eqn:thermo} We can further compute
\begin{align*}
\frac{dE}{dt} &= J_{01}(t)(E_1(t) - E_0(t))=k_BTJ_{01}(t)\log\frac{k_{10}}{k_{01}}\
\\\frac{dW}{dt} &= J_{01}(t) W_{01}(t) = J_{01}(t)(\phi_0(t) - \phi_1(t))=k_BTJ_{01}(t)\log\frac{u_{01}(t)k_{10}}{u_{10}(t)k_{01}} \
\\\frac{dQ}{dt} &= J_{01}(t) Q_{01}(t) = J_{01}(t)(E_0 + \phi_0(t) - E_1 - \phi_1(t))= k_BT J_{01}(t)\log\frac{u_{01}(t)}{u_{10}(t)}
\\\frac{dS}{dt} &= J_{01}(t) S_{01}(t) = J_{01}(t)\log\frac{p_0(t)}{p_1(t)}\
\\\frac{dF}{dt} &= k_BT J_{01}(t) \log\frac{k_{10}p_1(t)}{k_{01}p_0(t)}
\end{align*}

\item Define \textbf{Total Entropy Production} $S_\text{tot}(t)$ to be the total entropy produced from time $0$ to time $t$. In other words, $S_\text{tot}(0)=0$ and 

\[\frac{dS_\text{tot}(t)}{dt} = \frac{1}{k_BT}\frac{dQ}{dt} + \frac{dS}{dt}.
\]
After simplification,

\begin{align}\label{eqn:EntropyProduction} 
\frac{dS_\text{tot}(t)}{dt} = (p_0(t) u_{01}(t) - p_1(t) u_{10}(t))\log\frac{p_0(t) u_{01}(t)}{p_1(t) u_{10}(t)} \geq 0
\end{align}
which is a statement of the second law of thermodynamics.

\item The following identity is immediate
\begin{align*}
	\frac{dW}{dt} = \frac{dF}{dt}+ \frac{dS_\text{tot}}{dt}
\end{align*}
and is another form of the first law.
\end{enumerate}

\subsubsection{Thermodynamic cost for rapid erasing of a reliable bit}
How much does it cost for rapid erasing of a reliable bit, with the cost function equal to $\Delta W$? We claim that it costs $k_B \log 2$. In particular, neither the reliability timescale $\tau_r$ nor the erasing timescale $\tau_e$ appear in this answer.

Suppose we can erase a $(\tau_r, \tau_e)$ bit for work $W$. First note that $\frac{dW}{dt}$ is a function of $k_{01}/k_{10}$, $u_{01}(t)/u_{10}(t)$ and $J_{01}(t)$ as in (\ref{eqn:thermo}). In particular, simultaneously sending the rates $k_{01}$ and $k_{10}$ as low as possible while keeping their ratio the same has no effect on the work. So if we can erase a $(\tau_r, \tau_e)$ bit for work $W$, then we can erase a $(A\tau_r, \tau_e)$ bit for work $W$, for an arbitrarily large constant $A$. In particular it is enough for us to demonstrate a protocol when $\tau_r=1$.

Now note that $\frac{dW}{dt}$ depends only on the ratio $u_{01}(t)/u_{10}(t)$ and not on the actual values of the rates. We can also erase a $(\tau_r, \tau_e/2)$ bit for work $W$ by taking the $(\tau_r, \tau_e)$-protocol $(u_{ij}(t))$ and defining new rates $v_{ij}(t) = 2u_{ij}(2t)$. Since $v_{01}(t)/v_{10}(t) = u_{01}(2t)/ u_{10}(2t)$, it follows from a simple calculation that the work required does not change.

By taking a limit of this time scaling argument, we only need to erase a $(1,+\infty)$ bit. Here the infinite-time isothermal protocol, which proceeds by raising the `1' well infinitesimally, waiting for the system to equilibrate, and repeating, erases for a total work of $k_BT\log 2$ since that is the free energy difference between the initial and final state, and there is no extra dissipation. This establishes our claim.

A more detailed version of this calculation can be found in \cite{Browne2014Guaranteed}. This work assumes that there is a maximum energy limit $E_\text{max}$ to which a state can be raised, so that there will be some small error to erasing. It also makes another assumption about thermalization timescale which translates in our setting to assuming that there is a maximum value to the rates $u_{01}(t)$ and $u_{10}(t)$. With these assumptions, they show that the cost of rapid erasing  is slightly more than $k_B \log 2$ and goes to $k_B T \log 2$ very quickly as the timescale of thermal relaxation becomes smaller and $E_\text{max}$ goes to infinity.

\subsubsection{Link between KL-cost and thermodynamic work}
We will now characterize entropy production in terms of time reversal. This will allow us to make a link between KL-cost and the thermodynamic work $W$. 

We first recall the notion of \textbf{time reversal} of a Markov chain. Usually time-reversal is defined for time-homogeneous Markov chains. However, for the purposes of characterizing entropy production in terms of time reversal, we will work with a definition that applies to time-inhomogeneous Markov chains also. Instead of giving this definition in full generality, we work with a Markov chain with a finite state space. This is sufficient for our purposes, and allows us to avoid dealing with certain technical issues.

Note that given a matrix $U$ with positive non-diagonal entries and $U\cdot 1 = 0$, there is a \textbf{nonnegative} vector $v$ such that $v\cdot U = 0$. This can be shown by applying the Perron-Frobenius theorem to the exponential matrix $\e^U$.

\begin{Definition}[Time-reversal]\label{def:timereversal}
Consider a continuous-time time-inhomogeneous Markov chain with state space $[n] = \{1,2,\dots,n\}$ described by a time-dependent transition matrix $U(t) = (u_{ij}(t))_{i,j\in[n]}$ (so that at time $t$, $u_{ij}(t)$ denotes the rate of jumping to state $j$ given that the system is in state $i$). Let $\pi(t)$ be a sequence of stationary probability distributions on $[n]$, i.e., $\pi(t) U(t) = 0$ and $\pi_i\geq 0$ for all $i\in[n]$ and $\sum_i \pi_i = 1$. Then the \textbf{time-reversal} Markov chain is described by the time-dependent transition matrix $\hat{U}(t) = (\hat{u}_{ij}(t))_{i,j\in[n]}$ where

\[
 \hat{u}_{ij}(t) = \frac{\pi_i(t) u_{ij}(t)}{\pi_i(t)}
\] 
A Markov chain is \textbf{reversible} iff $\hat{U} = U$.
\end{Definition}

The justification for considering time-reversal comes from Bayes' rule. Reversible Markov chains are well-known to be characterized by the conditions of existence of a detailed balanced equilibrium, as well as by the Kolmogorov chain conditions. In particular, two-state Markov chains are always reversible.

For the special case of Equation~\ref{eqn:control} in particular, given a distribution $q$ at time $\tau_e$, the time reversal Markov chain evolves in time according to the ODE:

\begin{equation}
\begin{aligned}\label{eqn:revcontrol}
\dot{q}_0(t) &=  - u_{01}(t)q_0(t) + u_{10}(t)(1-q_0(t)),\
\\ q(\tau_e) &= q.
\end{aligned}
\end{equation}
We see that the difference is that in Equation~\ref{eqn:control}, the boundary condition was specified at time $0$, whereas here the boundary condition is specified at time $\tau_e$.

We define the \textbf{time-reversed measure} $\mu_{u,q}^\text{rev}$ as the measure on path space corresponding to the process described by Equation~\ref{eqn:revcontrol}. Strictly speaking, we should write $\mu_{u,q,\tau_e}^\text{rev}$ to denote the time at which the boundary condition is provided to the differential equation, but we will avoid this by using the convention that we're always going to set the boundary condition at time $\tau_e$ when considering the time-reversal.

The following result is key to our comparison.

\begin{Theorem}\label{thm:irrevdiss}
Run the control dynamics Equation~\ref{eqn:control} forward from initial condition $p(0)$ upto time $\tau_e$ to obtain the distribution $p(\tau_e)$. Consider the measure $\mu_{u,p(\tau_e)}^\text{rev}$. Then the total entropy production $S_\text{tot}(\tau_e)$ from time $0$ to time $\tau_e$ equals 

\[
S_\text{tot}(\tau_e) = D(\mu_{u,p(0)} \|\mu_{u,p(\tau_e)}^\text{rev}).
\]
\end{Theorem}
\begin{proof}
We will show that the time derivative of the RHS equals the rate of entropy production. This will prove the theorem.

Fix a time $t\in[0,\tau_e]$. Let the probability distribution at time $t$ be represented by $p(t)=(p_0(t), p_1(t))$. Let $f_{ij}:= p_i(t) u_{ij}(t)$ denote the flow rate from state $i$ to state $j\neq i$ at time $t$. Then $p_i(t+h) = p_i(t) + h(f_{ji} - f_{ij}) + o(h)$ where $i\neq j$ and $o(h)$ denotes terms $g(h)$ such that $\lim_{h\to 0} g(h)/h \to 0$.

We will consider the probabilities of the four Markov chain transitions $0\to 0$, $0\to 1$, $1\to 0$ and $1\to 1$ in the interval $(t,t+h)$ in the limit $h\to 0$ according to $\mu_{u,p(0)}$ and according to $\mu_{u,p(\tau_e)}^\text{rev}$. Up to terms of size $o(h^2)$, we have for $i\neq j$:

\begin{align*}
\mu_{u,p(0)}(i\to j) &= h p_i(t) u_{ij}(t)=hf_{ij}\
\\\mu_{u,p(\tau_e)}^\text{rev}(i\to j) &= h p_j(t+h) u_{ji}(t) = h(f_{ji} + h(f_{ij}-f_{ji}))
\end{align*}
The increment in the relative entropy in the time interval $(t,t+h)$ equals, upto $o(h^2)$ terms:

\begin{align*}
	&(1-hf_{01})\log\frac{1-hf_{01}}{ 1-hf_{10} -h^2(f_{01}-f_{10})} + hf_{01}\log\frac{hf_{01}}{hf_{10} + h^2(f_{01}-f_{10})}\
	\\&+ hf_{10}\log\frac{hf_{10}}{hf_{01} + h^2(f_{10}-f_{01})} + (1 - hf_{10})\log\frac{1 - hf_{10}}{1-hf_{01} - h^2(f_{10}-f_{01})}
\end{align*}
The off-diagonal terms contribute:

\begin{align*}
	&(1-hf_{01})\log\frac{1-hf_{01}}{ 1-hf_{10} -h^2(f_{01}-f_{10})}+(1 - hf_{10})\log\frac{1 - hf_{10}}{1-hf_{01} - h^2(f_{10}-f_{01})}\
	\\&\approx(1 + o(h))(hf_{10}-hf_{01} + o(h^2)) + (1 + o(h))( hf_{01} - hf_{10} + o(h^2))\
	\\& \approx o(h^2) 
\end{align*}
Divide by $h$, and take the limit $h\to 0$. We can ignore the off-diagonal terms. The diagonal terms sum to the rate of entropy production $\frac{dS_{\text{tot}}}{dt}$ as in Equation~\ref{eqn:EntropyProduction}, and we are done.
\end{proof}

By the First Law of Thermodynamics and Theorem~\ref{thm:irrevdiss},

\begin{align}\label{eqn:firstlaw}
	\Delta W = \Delta F + k_B T D(\mu_{u,p}\| \mu_{u,p(\tau_e)}^\text{rev})
\end{align}
where the increase in free energy of the system
\[\Delta F = k_B T\left(D(p(\tau_e)\|\pi)- D(p(0)\| \pi)\right)\] by Equation~\ref{SLP}. 
Now to compare our cost function with $\Delta W$. 

\begin{Theorem}
The $KL$-cost equals change in free energy by $k_B T$ plus a path-space relative entropy term that resembles entropy production:

\begin{align}\label{eqn:newfirstlaw}
k_B T D(\mu_{u,p}\,\|\, \mu_{k,p}) = \Delta F + k_B T D(\mu_{u,p} \| \mu_{k,p(\tau_e)}^\text{rev})
\end{align}
where $p(\tau_e)$ is --- as in Equation~\ref{eqn:firstlaw} --- the solution to the control dynamics Equation~\ref{eqn:control} at time $\tau_e$. 
\end{Theorem}
\begin{proof}
Using Equation~\ref{SLP}, we can rewrite the claim as

\[
D(\mu_{u,p}\,\|\, \mu_{k,p}) + D(p(0)\| \pi)=  D(\mu_{u,p} \| \mu_{k,p(\tau_e)}^\text{rev}) + D(p(\tau_e)\| \pi)
\]
Both LHS and RHS equal $D(\mu_{u,p}\,\|\, \mu_{k,\pi})$. The assertion for the LHS is straightforward. The assertion for the RHS is true because time-reversal dynamics was defined to keep the stationary distribution $\pi$ remain stationary under time reversal.
\end{proof}

Comparing \ref{eqn:firstlaw} and \ref{eqn:newfirstlaw}, a KL control treatment replaces the total entropy production $D(\mu_{u,p}\| \mu_{u,p(\tau_e)}^\text{rev})$ in \ref{eqn:firstlaw} by the new term $D(\mu_{u,p} \| \mu_{k,p(\tau_e)}^\text{rev})$ which compares the control dynamics with the time reversal of the passive dynamics. This suggests an interpretation as follows. If we applied the control during some time interval $[0,\tau_e]$, and remembered what control we applied, then the entropy production is correctly given by $D(\mu_{u,p}\| \mu_{u,p(\tau_e)}^\text{rev})$. However, the information that a control was applied also needs to be stored somewhere. If we forget that a control was applied, and if application of the control is very rare, then our default model for the dynamics should be much closer to the passive dynamics. In this case, entropy production may be closer to the value $D(\mu_{u,p} \| \mu_{k,p(\tau_e)}^\text{rev})$.


\subsection{Large deviations interpretation} Our cost function $D(\mu\|\nu)$ also admits a large deviation interpretation which was, remarkably, already noted by Schr{\"o}dinger in 1931~\cite{schrodinger1931uber,beurling1960automorphism,follmer1988random,aebi1996schrodinger}. Motivated by quantum mechanics, Schr{\"o}dinger asked: conditioned on a more or less astonishing observation of a system at two extremes of a time interval, what is the least astonishing way in which the dynamics in the interval could have proceeded? Specializing to our problem of erasing, suppose an ensemble of two-state Markov chain with passive dynamics given by Equation~\ref{eqn:passive} was observed at time $0$ and at time $\tau_e$. Suppose the empirical state distribution over the ensemble was found to be the equilibrium distribution $\pi$ at time $0$, and $(1, 0)$ at time $\tau_e$ respectively. This would be astonishing because no control has been applied, yet the ensemble has arrived at a state of higher free energy. Conditioned on this rare event having taken place, what is the least unlikely measure $\mu^*$ on path space via which the process took place?
	
By a statistical treatment of multiple single particle trajectories, Schr{\"o}dinger found that the likelihood of an empirical 
measure $\mu$ on path space falls exponentially fast with the relative entropy $D(\mu \| \nu)$ where $\nu$ is the measure induced 
by the passive dynamics. In particular, the least unlikely measure $\mu^*$ is that measure which --- among all $\mu$ whose 
marginals at time $0$ and time $\tau_e$ respect the observations --- minimizes $D(\mu \| \nu)$. So for the problem of erasing 
where $k_{01}=k_{10}$, the measure $\mu$ varies over all measures that have marginal $(1/2,1/2)$ at time $0$ and marginal $(1,0)$ 
at time $\tau_e$, and $\mu^*$ is that measure among all such $\mu$ that minimizes $D(\mu\| \mu_{k,(1/2,1/2)})$. Thus our optimal 
control produces in expectation the least surprising trajectory among all controls that perform rapid erasing.


\subsection{Gibbs measure} Equation~\ref{eqn:gibbs} is not accidental for this example, but is in fact a general feature when the cost function is relative entropy~\cite{dupuis2011weak}. More abstractly, the Radon-Nikodym derivative (i.e., ``probability density'') $\frac{d\,\mu^*}{d\,\nu}$ of the measure $\mu^*$ induced on path space by the optimal control $u^*$ with respect to the measure $\nu$ induced by the passive dynamics is a \textbf{Gibbs measure}, with the cost-to-go function $v(t)$ playing the role of an energy function. In other words, mathematically our problem is precisely the free energy minimization problem so familiar from statistical mechanics. There is also a possible physical interpretation: we are choosing paths in $\mathcal{P}$ as microstates, instead of points in phase space. The idea of paths as microstates has occurred before~\cite{wissner2013causal}.

\section{Concluding remarks}
Since charging a battery can also be thought of as erasing a bit~\cite{gopalkrishnan2013hot}, our result may also hold insights into the limits of efficiencies of rapidly charging batteries that must simultaneously hold their energy for a long time.

So long as the noise is Markovian, we conjecture that the KL cost for erasing the two-state Markov chain is a lower bound for more general cases -- for example for bits with Langevin dynamics~\cite{zwanzig2001nonequilibrium}, which is a stochastic differential equation expressing Newton's laws of motion with Brownian noise perturbations.

\vspace{6pt} 
\acknowledgments{\textbf{Acknowledgments:} I thank Sanjoy Mitter, Vivek Borkar, Nick S. Jones, Mukul Agarwal, and Krishnamurthy Dvijotham for helpful discussions. I thank Abhishek Behera for drawing Figure~\ref{fig1}.}


\bibliographystyle{mdpi}
\renewcommand\bibname{References}

%
%
%

\end{document}